\documentclass[10pt]{article}
\usepackage[T1]{fontenc}
\usepackage{tcolorbox}

\usepackage[margin=1in]{geometry}
\usepackage{amsmath,array,bm}
\usepackage{amsmath,amsfonts}
\usepackage{amsthm}
\usepackage{amssymb,latexsym}
\usepackage{algorithm}
\usepackage{asymptote}
\usepackage{subcaption}
\usepackage{algorithmicx}
\usepackage{marginnote}
\usepackage[]{algpseudocode}
\usepackage{float}
\usepackage{thmtools}
\usepackage{mathtools}
\usepackage{thm-restate}
\usepackage{bussproofs}
\usepackage[square,sort,comma,numbers]{natbib}
\usepackage{tikz}
\usepackage{xspace}
\usepackage{xcolor}
\usepackage{tipa}
\usepackage{mdframed}
\usepackage[export]{adjustbox}[2011/08/13]
\usepackage{multirow}
\usepackage{hhline}
\usepackage{todonotes}
\usepackage{tcolorbox} 
\usepackage{dirtytalk}
\usepackage{yhmath}
\usepackage{enumitem}
\usepackage[pdftex, plainpages = false, pdfpagelabels, 
                 bookmarks=false,
                 bookmarksopen = true,
                 bookmarksnumbered = true,
                 breaklinks = true,
                 linktocpage,
                 pagebackref,
                 colorlinks = true,  
                 linkcolor = blue,
                 urlcolor  = blue,
                 citecolor = red,
                 anchorcolor = green,
                 hyperindex = true,
                 hyperfigures
                 ]{hyperref}

\newcommand{\cd}{\mathcal{D}}

\newtheorem{lemma}{Lemma}

\newtheorem{observation}{Observation}
\newtheorem{theorem}{Theorem}

\newtheorem{proposition}{Proposition}

\algtext*{EndWhile}
\algtext*{EndIf}
\algtext*{EndFor}

\title{\bf Extraction Theorems With Small Extraction Numbers}

\author{
Arjun Agarwal\thanks{Jesuit High School, USA. Email: \texttt{arjunagarwal010@gmail.com}.}\ \ \ \ \ 
Sayan Bandyapadhyay\thanks{Portland State University, USA. Email: \texttt{sayanb@pdx.edu}.}
} 
\date{} 

\theoremstyle{definition}

\begin{document}
\maketitle

\begin{abstract}
In this work, we develop Extraction Theorems for classes of geometric objects with small extraction numbers. These classes include intervals, axis-parallel segments, axis-parallel rays, and octants. We investigate these classes of objects and prove small bounds on the extraction numbers. The tightness of these bounds is demonstrated by examples with matching lower bounds.
\end{abstract} 


\section{Introduction}
Designing approximation algorithms for geometric covering problems is a popular line of research in computational geometry \cite{MustafaR09,HochbaumM85,dcg/BandyapadhyayBI20,ErlebachL08,mmrs-16}. In a typical covering problem, given a set of points and a set of geometric objects, the goal is to pick a subset of objects whose union contains the input points (i.e., cover them) and the subset optimizes a certain objective function. Covering problems arise in fields such as wireless and sensor networks, VLSI design, robotics, biology, and image processing~\cite{abrams2004set,cardei2005energy,chan2004hiding,HochbaumM85,huang2007predicting,choset2001coverage,svennebring2004building}. In recent time, an interesting byproduct of this research is the Extraction Theorem for geometric objects \cite{DBLP:conf/innovations/BandyapadhyayMR24}. For any set of weighted objects $S$, let $W(S)$ denote the sum of the weights of the objects.

\begin{theorem}[Extraction Theorem for Disks]
\label{thm:extract-disks}
\cite{DBLP:conf/innovations/BandyapadhyayMR24} Suppose we are given a set $\cd$ of disks along with a weight function $w:\cd \rightarrow \mathbb{R}^+$ and a set $T$ of $n$ points in the plane, such that each point of $T$ is contained in at least two disks of $\cd$. Then there exists a subset $\mathsf{sol}\subset \cd$ such that $\mathsf{sol}$ covers $T$ and $W(\cd\setminus \mathsf{sol})\ge W(\cd)/4$. Moreover, such a subset $\mathsf{sol}$ can be computed in polynomial time.
\end{theorem}

This fundamental result led to the design of the first constant-factor approximation algorithms for a number of covering problems with the approximation factor of 4 \cite{DBLP:conf/innovations/BandyapadhyayMR24}. Intuitively, the theorem states that a subset of objects of weight of at least a quarter of the total weight can be removed (or extracted) and the residual objects still cover all the points. Note that the requirement that each point is in at least two objects is necessary. Moreover, the factor $1/4$ is tight: consider an arrangement of 4 disks where each pair contains a unique point -- 3 disks are needed in any cover. We refer to the factor 4 in Theorem \ref{thm:extract-disks} as an \textit{extraction number} for disks. In general, we say the Extraction Theorem holds for a class of objects with an \textit{extraction number} of $\alpha$ if for any set $\mathcal{O}$ of these objects with a weight function $w:\mathcal{O} \rightarrow \mathbb{R}^+$ and a set $T$ of $n$ points, such that each point of $T$ is contained in at least two objects of $\mathcal{O}$, there exists a subset $\mathsf{sol}\subset \mathcal{O}$ such that $\mathsf{sol}$ covers $T$ and $W(\mathcal{O}\setminus \mathsf{sol})\ge W(\mathcal{O})/\alpha$. 

Note that the lesser the extraction number, the better, as we will need smaller per unit weight to cover the points. Thus we are interested in the smallest possible extraction number for a given class of objects. However, it might not be possible to have a constant extraction number for all classes. A broad class for which this is known to be possible is the class of objects having linear union complexity \cite{DBLP:conf/innovations/BandyapadhyayMR24}. For rectangles, the best-known upper bound on the minimum extraction number is $O(\log n)$ \cite{DBLP:conf/innovations/BandyapadhyayMR24} which follows from \cite{ene2017geometric}. 

In this work, we investigate various classes with the focus of proving small bounds on the extraction numbers. In particular, we study the following classes: intervals in 1D, axis-parallel segments in 2D, axis-parallel rays in 2D, and octants in 3D containing $(+\infty, +\infty, +\infty)$. We prove the Extraction Theorem for each class with an extraction number $i$ where $i \in \{2,3,4\}$. (See Table \ref{table:1}.) We also show the tightness of these bounds by constructing examples with matching lower bounds.  

\begin{center}
\begin{table}[ht]
\centering
\begin{tabular}{ |p{6.1cm}|p{3.4cm}| }
 \hline
\centering \textbf{Objects $(\mathcal{O})$} &\textbf{Extraction number}\\ 
 \hline
\centering Intervals  &2 \\
\hline
\centering Axis-parallel segments  & 4 \\
\hline
\centering Axis-parallel rays of $i$ types for $i\in [2,4]$ &  $i$ \\
\hline
\centering Octants  & 4 \\
\hline
\centering Pseudo disks  & 4~\cite{smorodinsky2007chromatic}\\
\hline
\centering Bottomless rectangles  & 3~\cite{keszegh2012coloring} \\
\hline
\centering Half planes  & 3~\cite{keszegh2012coloring} \\
\hline
\centering Translates of bottomless rectangles 
& 2
~\cite{cardinal2020colouringbottomlessrectanglesarborescences}  \\
\hline
\centering Translates of quadrants
& 2 
~\cite{cardinal2020colouringbottomlessrectanglesarborescences}  \\
\hline
\end{tabular}
\caption{Summary of our results along with extraction number of objects that follow from Proposition~\ref{prop:coloring-to-extr-thm} and known results for coloring number of their hypergraphs.}
\label{table:1}
\end{table}
\end{center}
\section{Preliminaries}
Our results are based on a connection to geometric hypergraph coloring. Consider a hypergraph $H=(V,\mathcal{E})$. A $k$-coloring for $k\in \mathbb{N}$ of $H$ (i.e., the vertices) is a function $\phi: V\rightarrow \{1,\ldots,k\}$. A $k$-coloring $\phi$ of $H$ is called \textit{proper} if for every hyperedge $e\in \mathcal{E}$ with $|e|\ge 2$, $e$ contains $v_i,v_j$ with $\phi(v_i)\ne \phi(v_j)$, i.e., $e$ is non-monochromatic. The proper coloring number of $H$, denoted by $\chi(H)$, is the minimum integer $k$ for which there is a proper $k$-coloring of $H$. 

As mentioned before, we focus on hypergraphs induced by objects. Consider a set of objects $\mathcal{O}$ in $\mathbb{R}^d$. For a point $p$, let $o(p) = \{O\in \mathcal{O}\mid p\in O\}$. The hypergraph $H(\mathcal{O})$ induced by $\mathcal{O}$ is defined as $(\mathcal{O},\{o(p)\}_{p\in \mathbb{R}^d})$. Note that $H(\mathcal{O})$ has one hyperedge for each cell in the arrangement of $\mathcal{O}$. We establish the following connection between Extraction Theorems and proper hypergraph coloring. 

\begin{proposition}\label{prop:coloring-to-extr-thm}
    Consider any arbitrary set of objects $\mathcal{O}$ belonging to a specific class $\mathcal{C}$. Suppose $H(\mathcal{O})$ has a proper $\kappa$-coloring that can be computed in polynomial time. Then the Extraction Theorem holds for $\mathcal{C}$ with an extraction number of $\kappa$. 
\end{proposition}

\begin{proof}
    Consider any set $\mathcal{O}$ of objects from $\mathcal{C}$ and a set $T$ of points, such that each point of $T$ is contained in at least two objects of $\mathcal{O}$. Also, fix a proper $\kappa$-coloring $\phi$ of $H(\mathcal{O})$ computed in polynomial time. Let $i$ be a number such that the weight of the objects in $\phi^{-1}(i)$ is the maximum. This weight is at least $W(\mathcal{O})/\kappa$. Construct $\mathsf{sol}$ to be $\mathcal{O}\setminus \phi^{-1}(i)$. We claim that $\mathsf{sol}$ is a cover. Consider any $p\in T$. We know $p$ is contained in at least two objects of $\mathcal{O}$. As $p$ is in a cell of the arrangement of $\mathcal{O}$, there is a hyperedge $e=o(p)$ in $H(\mathcal{O})$. Also, $\phi$ is a proper coloring, and so $\exists j\ne i$ and $O\in e$ such that $\phi(O)=j$. Hence, $p$ is covered by $\mathsf{sol}$, as $p\in O$ and $O\in \mathsf{sol}$. Now, $W(\mathcal{O}\setminus \mathsf{sol})=W(\phi^{-1}(i))\ge W(\mathcal{O})/\kappa$. Hence, the proposition follows. 
\end{proof}

Proposition \ref{prop:coloring-to-extr-thm} gives a framework for proving Extraction Theorems for general objects. Proper coloring of geometric hypergraphs is a well-studied area and the existing results directly lead to Extraction Theorems for a wide range of objects. For example, for a set $\cd$ of disks, $\chi(H(\cd))=4$, and a 4-coloring can be computed in polynomial time \cite{smorodinsky2007chromatic}. This gives an alternative proof of Theorem \ref{thm:extract-disks}. It is also known that for pseudo-disks, a 4-coloring can be computed in polynomial time~\cite{keszegh2020coloring}. Hence, the same theorem follows even for pseudo-disks. Similarly, for $n$ Jordan regions with union complexity $\mathcal{U}(n)$, an $O(\mathcal{U}(n))/n$-coloring can be computed in polynomial time \cite{smorodinsky2007chromatic}. (See the survey~\cite{smorodinsky2013conflict} for more details.) When $\mathcal{U}(n)$ is linear, this gives another proof of the Extraction Theorem for the linear union complexity objects. Keszegh~\cite{keszegh2012coloring} showed that the proper coloring number of hypergraphs induced by bottomless rectangles and half-planes is 3. 
Table~\ref{table:1}
lists additional extraction numbers that follow from Proposition~\ref{prop:coloring-to-extr-thm} and known results for the coloring number of hypergraphs.
To obtain our results, we prove similar bounds for the concerned classes of objects. Since the objects correspond to the vertices of the hypergraph, for the remainder of the paper we will refer to a coloring of a hypergraph as a coloring of the geometric objects. 

\section{Extraction theorems for restricted classes}
\subsection{Intervals}\label{1D}

\begin{lemma}\label{lem:intervals}
    Given a set $\mathcal{I}=\{I_1,I_2,\ldots,I_m\}$ of intervals, a proper 2-coloring of $H(\mathcal{I})$ can be computed in polynomial time. 
\end{lemma}

\begin{proof}
We assume that the $m$ intervals lie on the 
$x$-axis and their union defines a connected region on the $x$-axis. Otherwise, we can consider the disconnected regions separately. 
Let the segment $I_i \in \mathcal{I}$
be represented by the ordered tuple $(a_i, b_i)$ denoting the $x$ co-ordinates 
of its endpoints with $a_i < b_i$. First, we compute a set of \textit{key} intervals $\mathcal{I}'=\{I_{i^1},I_{i^2},\ldots,I_{i^\tau}\}\subseteq \mathcal{I}$, one in every iteration. Initially, $I_{i^1}$ is the interval with the minimum $a_i$-value. If there are multiple such intervals, choose the longest one. Suppose we have already picked $\{I_{i^1},I_{i^2},\ldots,I_{i^j}\}$. $I_{i^{j+1}}$ is the interval that starts in $[a_{i^j},b_{i^j}]$, finishes strictly after $b_{i^j}$, and finishes the last, i.e., having the maximum $b_i$-value. If there is no such interval, stop. If multiple such intervals finish last, pick any of them.  

\begin{observation}
    For $2\le j\le \tau-1$, the only interval of $\mathcal{I}'$ that $I_{i^j}$ intersects are $I_{i^{j-1}}$ and $I_{i^{j+1}}$. 
\end{observation}

We color the intervals in the following way. $I_{i^1}$ is colored by color 1. Now, consider $I_{i^{j+1}}$ for $j\ge 1$. If there is an interval in $\mathcal{I}\setminus \mathcal{I}'$ that contains $I_{i^{j}}\cap I_{i^{j+1}}$, $I_{i^{j+1}}$ is colored by the color of $I_{i^{j}}$. Otherwise, $I_{i^{j+1}}$ is colored by the color distinct from that of $I_{i^{j}}$. Now, consider any interval $I_i\in \mathcal{I}\setminus \mathcal{I}'$. If $I_i$ is fully contained in $I_{i^{j}}\cap I_{i^{j+1}}$ for some $j$, color it arbitrarily. Otherwise, if $I_i$ is fully contained in an $I_{i^{j}}$ for some $j$, color it by the color distinct from that of $I_{i^{j}}$. Otherwise, $I_i$ is not fully contained in any of the intervals of $\mathcal{I}'$. In this case, $I_i$ fully contains $I_{i^{j}}\cap I_{i^{j+1}}$ for a unique $j$. Color $I_i$ by the color distinct from that of $I_{i^{j}}$.

We prove the coloring is proper. Consider any hyperedge $e$ corresponding to a point $p \in \mathbb{R}$ with $|e|\ge 2$. First, suppose $p\in I_{i^{j}}\cap I_{i^{j+1}}$ for some $j$. If $I_{i^{j}}$ and $I_{i^{j+1}}$ have distinct colors, we are done. So, suppose they have the same color. By our coloring scheme, there is an $I_i\in \mathcal{I}\setminus \mathcal{I}'$ that contains $I_{i^{j}}\cap I_{i^{j+1}}$. Moreover, there is a unique index $j$ for which this is true. By our coloring scheme, the color of $I_i$ is different from that of $I_{i^{j}}$, and $e$ is non-monochromatic. 

Now, suppose $p$ is in a unique $I_{i^{j}}$. Thus, there is an $I_i\in \mathcal{I}\setminus \mathcal{I}'$ that contains $p$. Also, this $I_i$ may be fully contained in at most one interval of $\mathcal{I}'$. If it is fully contained in an interval of $\mathcal{I}'$, that must be $I_{i^{j}}$, as both intervals contain $p$. So, the color of $I_i$ is different from that of $I_{i^{j}}$. Otherwise, $I_i$ fully contains either $I_{i^{j-1}}\cap I_{i^{j}}$ or $I_{i^{j}}\cap I_{i^{j+1}}$. In both cases, the color of $I_i$ is different from that of $I_{i^{j}}$ by our scheme. 
\end{proof}

\begin{theorem}
    Intervals admit an extraction number of 2 and there are instances with intervals where the minimum extraction number is 2. 
\end{theorem}
\begin{proof}
    Proposition \ref{prop:coloring-to-extr-thm} and Lemma \ref{lem:intervals} imply that  intervals admit an extraction number of 2. Consider the case where we have two unit-weight intervals with a common point. Any cover must pick one interval demonstrating that there are instances where the minimum extraction number is $2$. 
\end{proof}

\subsection{Axis-parallel segments}\label{2Dseg}

\begin{lemma}\label{lem:2Dseg}
    Suppose $\mathcal{X} = \cup_{i=1}^{m_1}\mathcal{X}_i$ where each $\mathcal{X}_i$ is a set of horizontal segments with the same $y$ coordinate value and $\mathcal{Y} = \cup_{i=1}^{m_2}\mathcal{Y}_i$ where each $\mathcal{Y}_i$ is a set of vertical segments with the same $x$ coordinate value. For $\mathcal{D} = \mathcal{X} \cup \mathcal{Y}$, with $|D| = m_1 + m_2 = m$, a proper 4-coloring of $H(\mathcal{D})$ can be computed in polynomial time. 
\end{lemma}
\begin{proof}
    Each of the sets $\mathcal{X}_i$ and $\mathcal{Y}_j$ is a set of intervals in 1D. 
    Thus, as shown in Lemma~\ref{lem:intervals} a proper 2-coloring of $H(\mathcal{X}_i)$ and $H(\mathcal{Y}_j)$ can be computed in polynomial time. Color each $H(\mathcal{X}_i)$ using colors $\{1,2\}$ and $H(\mathcal{Y}_j)$ using colors $\{3,4\}$. 
    Now, the intersection of any pair of segments from $\mathcal{X}_i$ and $\mathcal{X}_{i'}$, respectively, is empty. The same holds between $\mathcal{Y}_{j}$ and $\mathcal{Y}_{j'}$. The only other intersections are horizontal-vertical intersections across $\mathcal{X}_i$ and $\mathcal{Y}_j$. Therefore $H(\mathcal{D})$  will consist of copies of $H(\mathcal{X}_i)$ for $1\le i\le m_1$, $H(\mathcal{Y}_j)$ for $1\le j\le m_2$, plus additional hyperedges corresponding to intersection points across $\mathcal{X}_i$ and $\mathcal{Y}_j$. Each copy of $H(\mathcal{X}_i)$, $H(\mathcal{Y}_j)$ has a proper $2$-coloring. In addition, the new hyperedges connect vertices that have different colors. Therefore we have a proper $4$-coloring of $H(\mathcal{D})$. 
\end{proof}
\begin{theorem}
    Axis-parallel segments admit an extraction number of 4 and there are instances with axis-parallel segments where the minimum extraction number is 4. 
\end{theorem}
\begin{proof}
By Proposition~\ref{prop:coloring-to-extr-thm} and Lemma~\ref{lem:2Dseg}, we obtain the Extraction Theorem for axis-parallel segments in 2D with an extraction number of 4. 
We will now prove that there are instances with axis-parallel segments where the minimum extraction number is $4$. 

Let us assume that the assertion is not true and the minimum extraction number for axis-parallel segments is $4-\epsilon$ for an absolute constant $\epsilon$. This implies that for a set $\mathcal{O}$ of segments and a set of points $T$, there is a cover $\mathsf{sol}$ such that $W(\mathcal{O} \setminus \mathsf{sol}) \geq W(\mathcal{O})/(4-\epsilon)$, and so $W(\mathsf{sol}) \leq W(\mathcal{O}) - \dfrac{W(\mathcal{O})}{4-\epsilon}$. 
\begin{figure}
    \centering
    \includegraphics[width=0.5\linewidth]{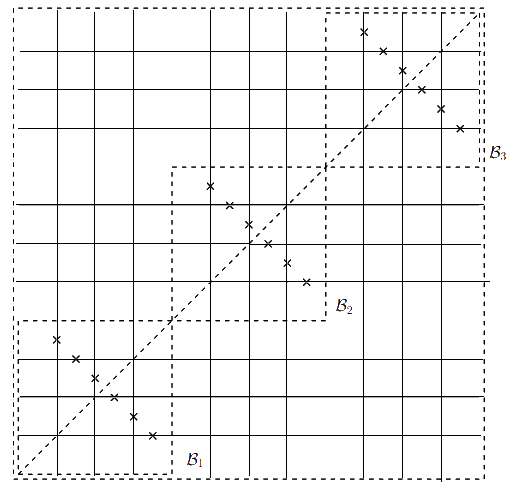}
    \caption{Configuration reproduced from \cite{coaduroindependence} demonstrating tightness of extraction number}
    \label{fig:MK}
\end{figure}
Consider the example instance in Figure~\ref{fig:MK}    with $m$ axis-parallel segments arranged as follows. The smaller box, labeled as $\mathcal{B}_i$, in general, consists of $2k$ horizontal and $2k$ vertical segments distributed on $k$ horizontal and $k$ vertical lines, respectively. This box is referred to as a $k$-box. The crosses mark the intersection point of these segments. For the crosses on the vertical lines, there is one up segment and one down segment that meet at the intersection point. Similarly, for the crosses on the horizontal line, there is one left segment and one right segment meeting at the intersection point. In Figure~\ref{fig:MK}, there are three 3-boxes. Consider an arrangement of a larger square with $k$ distinct $k$-boxes arranged in a diagonal similar to Figure~\ref{fig:MK}. Such an arrangement has $m=4k^2$ axis parallel segments.
Coaduro~et al.~\cite{coaduroindependence} proved that the independence number of the intersection graph of these segments is at most $\dfrac{m}{4}+c \sqrt{m}$  for some constant $c$. Considering weight 1 for each segment, this implies that the weight of any cover $\mathsf{sol}$ of all the intersection points of these segments,   
$W(\mathsf{sol}) \geq \dfrac{3m}{4} - c \sqrt{m}$. But, we know that $W(\mathsf{sol}) \leq \dfrac{m}{4-\epsilon}$ due to our assumption on the extraction number. Combining the two inequalities for $W(\mathsf{sol})$, we have $\dfrac{3m}{4} - c \sqrt{m} \leq m - \dfrac{m}{4-\epsilon}$. Solving this inequality gives $\epsilon \leq \dfrac{16c}{4c+\sqrt{m}}$. Now $\epsilon$ can be made arbitrarily small by increasing $m$. Therefore, it cannot be an absolute constant. This proves that the bound of $4$ for the extraction number is tight.
\end{proof}

\subsection{Axis-parallel rays}

A ray can be considered as an extension of a line segment with one of the sides open (unbounded). We define rays as of orientation 1, 2, 3, 4 depending on their orientation as follows. Rays of orientation 1 are parallel to the $x$ axis and open at the right in the direction of the positive $x$ axis. Rays of orientation 2 are parallel to the $x$ axis and open at the left in the direction of the negative $x$ axis. 
Rays of orientation 3 are parallel to the $y$ axis and open at the top in the direction of the positive $y$ axis. Rays of orientation 4 are parallel to the $y$ axis and open at the bottom in the direction of the negative $y$ axis.
We call a set of rays a configuration of type $i$ if it contains rays of $i$ types. For example, a configuration of type 3 could contain rays of orientation 1, 2 and 4. 
 
\begin{lemma}\label{lem:type2ray}
    Suppose $\mathcal{R} = \{R_1, R_2, \dots, R_m \}$ is a set of rays of type 2 where the set can be split into two subsets of rays $\mathcal{R}_1, \mathcal{R}_2$ with rays in each subset of same orientation. A proper 2-coloring of $H(\mathcal{R})$ can be computed in polynomial time. 
\end{lemma}

\begin{proof}
    For all subsets of overlapping rays in $\mathcal{R}_1$ select the \textit{dominating} ray, which will be represented by the extremal starting point, i.e., the dominating ray contains all other overlapping rays. Color these dominating rays with color 1. Assign all other rays in $\mathcal{R}_1$ a color 2. Similarly
    for all subsets of overlapping rays in $\mathcal{R}_2$ select the dominating ray, which will be represented by the extremal starting point,  and color it with color 2. Assign all other rays in $\mathcal{R}_2$ a color 1.

    We prove that that the coloring $\phi$ is proper. Consider a hyperedge $e$ corresponding to a point $p$ with $|e|\ge 2$. If the point $p$ is at the overlap of rays in $\mathcal{R}_1$ or $\mathcal{R}_2$, then by our coloring algorithm, there exists $v_i, v_j\in e$ such that $\phi(v_i) \neq \phi(v_j)$. If the point $p$ is at the intersection of two rays of different orientations, then it is also at the intersection of the respective dominating rays and the corresponding vertices have different colors. This finishes the Lemma.  
\end{proof}

By Proposition~\ref{prop:coloring-to-extr-thm} and Lemma~\ref{lem:type2ray} we obtain the Extraction Theorem with an extraction number of $2$. The bound 2 is tight. Consider a common point at the intersection of two rays, each of different orientations. At least one of the two rays is needed to cover the point.

\begin{lemma}\label{lem:type3ray}
   Suppose $\mathcal{R} = \{R_1, R_2, \dots, R_m \}$ is a set of rays of type 3 where the set can be split into three subsets, $\mathcal{R}_1, \mathcal{R}_2, \mathcal{R}_3$ with rays in each subset of same orientation. A proper 3-coloring of $H(\mathcal{R})$ can be computed in polynomial time. 
\end{lemma}
\begin{proof}
Without loss of generality, we assume that the rays in $\mathcal{R}_1, \mathcal{R}_2$ are parallel to the $x$-axis and rays in $\mathcal{R}_3$ are parallel to the $y$-axis. Lemma~\ref{lem:type2ray} shows that a proper 2-coloring of $H(\mathcal{R}_1 \cup \mathcal{R}_2)$ can be computed in polynomial time. Let us call these colors 1 and 2.
Consider the addition of rays from $\mathcal{R}_3$. Color the dominating rays in this set with color 3 and all other rays with any other color from 1 or 2. 

We prove that this coloring $\phi$ is proper. Consider a hyperedge $e$ corresponding to a point $p$ that is contained in 2 rays. If the point is at the overlap of rays in $\mathcal{R}_1 \cup \mathcal{R}_2$, then by Lemma~\ref{lem:type2ray} there exists $v_i, v_j\in e$ such that $\phi(v_i) \neq \phi(v_j)$. If the point $p$ is at the intersection of two rays, respectively from  $(\mathcal{R}_1 \cup \mathcal{R}_2)$ and $\mathcal{R}_3$, then $e$ contains a vertex of color 3 and another vertex of color $1$ or $2$. Therefore,
there exists $v_i, v_j\in e$ such that $\phi(v_i) \neq \phi(v_j)$. This finishes the Lemma.  
\end{proof}

By Proposition~\ref{prop:coloring-to-extr-thm} and Lemma~\ref{lem:type3ray} we obtain the Extraction Theorem with an extraction number of 3. In order to prove that $3$ is the lower bound,
consider a set of $3k$ rays containing exactly $k$ rays of orientations 1, 2, and 3 as shown in Figure \ref{fig:3-types} for $k=4$. Orientation 1 and 2 rays touch each other at a single point. We claim that the independence number of the corresponding intersection graph is at most $k+1$. Consider any independent set $I$. Let $i$ be the maximum index of an orientation 2 ray in $I$. Also, let $j$ be the minimum index of an orientation 1 ray in $I$. If $i \ge j$, then orientation 1 and 2 rays in $I$ intersect all rays of orientation 3, and thus $I$ contains only orientation 1 and 2 rays. But, due to their arrangement, $I$ contains at most $k$ rays. So, suppose $i < j$. Thus, there can be at most $i$ rays of orientation 2 and $k-j+1$ rays of orientation 1 in $I$. Now, the orientation 2 ray of index $i$ intersects the first $i$ orientation 3 rays. Also, the orientation 1 ray of index $j$ intersects the last $k-j$ orientation 3 rays. Thus, the number of orientation 3 rays in $I$ is at most $k-(i+k-j)=j-i$, and so the total number of rays in $I$ is at most $i+(k-j+1)+j-i=k+1$. It follows that the size of any cover of the intersection points of the rays is at least $2k-1$. Hence, the extraction number must be at least $3/(1+1/k)$, which tends to $3$ as $k$ increases. 

\begin{figure}[ht]
    \centering
    \includegraphics[width=0.2\linewidth]{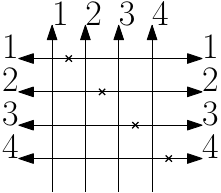}
    \caption{Figure showing a lower bound instance for type 3 rays.}
    \label{fig:3-types}
\end{figure}

\begin{lemma}\label{lem:type4ray}
    Suppose $\mathcal{R} = \{R_1, R_2, \dots, R_m \}$ is a set of rays where the set can be split into four subsets, $\mathcal{R}_1, \mathcal{R}_2, \mathcal{R}_3, \mathcal{R}_4$  with rays in each subset of same orientation. A proper 4-coloring of $H(\mathcal{R})$ can be computed in polynomial time. 
\end{lemma}
\begin{proof}
    Using a similar strategy as the proof of Lemma~\ref{lem:2Dseg}, we can compute a proper 4 coloring of $H(\mathcal{R})$ in polynomial time. This is true, as each ray can be treated as a segment 
    with extremal points on the rays acting as endpoints of the segment. 
\end{proof}

By Proposition~\ref{prop:coloring-to-extr-thm} and Lemma~\ref{lem:type4ray} we obtain the Extraction Theorem with an extraction number of 4. 
We note that the example instance from \cite{coaduroindependence} shown in Figure \ref{fig:MK} can be modified to convert the segments into rays. Specifically, the right, left, up, and down segments can be made to be rays of orientations 1, 2, 3, and 4, respectively. Hence, the bound 4 is tight even for rays.    

The above discussion can be combined into a single theorem for axis parallel rays of type $i$. 

\begin{theorem}\label{thm:typeiray}
Axis-parallel rays of type $i$ ($1 < i \leq 4$) admit an extraction number of $i$ and there are instances with axis-parallel rays of type $i$ where the minimum extraction number is $i$.
\end{theorem}

\subsection{Octants in 3D} 
\begin{lemma}\label{lem:octant}
    Given a set $\mathcal{O} = \{O_1, O_2, \dots, O_m\}$ of octants containing $(+\infty, +\infty, +\infty)$, a proper 4-coloring of $H(\mathcal{O})$ can be computed in polynomial time.
\end{lemma}
\begin{proof}
    For $1\le i\le m$, let $(a_i,b_i,c_i)$ be the apex of $O_i$. Thus, $O_i$ is defined as the unbounded region $\{x\ge a_i,y\ge b_i,z\ge c_i\}$. An octant $O_i$ is said to \textit{dominate} another octant $O_j$ if $O_i$ contains $O_j$, or equivalently if $a_j\ge a_i, b_j\ge b_i$, and $c_j\ge c_i$. Let $\mathcal{O}'$ be the maximal subset of $\mathcal{O}$ such that for any octant $O_i$ in $\mathcal{O}'$, there is no $O_j$ in $\mathcal{O}$ that dominates $O_i$. First, we compute $\mathcal{O}'$ by comparing the pairs of octants. 
For all pairs of octants $O_i,O_j$ in $O'$, let $c_{ij} = \max(a_i,a_j)+\max(b_i,b_j)+\max(c_i,c_j)$. Let $c_{\max} = \max_{i,j} c_{ij}$. 
Let $\mathcal{T}'$ be the set of triangles formed by projecting the octants of $O'$ on the plane $x+y+z = c_{max}$. 
\begin{observation}
 The triangles in $\mathcal{T}'$ are equilateral and oriented in the same direction. Therefore, they form a set of pseudo-disks.
\end{observation}
Since the set $\mathcal{T}'$ is a set of pseudo disks, a proper 4-coloring of the corresponding 
hypergraph $H(\mathcal{T}')$ can be computed in polynomial time \cite{keszegh2020coloring}. Color each octant in $\mathcal{O}'$ by the color of its projection in $\mathcal{T}'$. 
By the maximality of $\mathcal{O}'$, for any $O_i$ in $\mathcal{O} \setminus \mathcal{O}'$, there is an $O_j$ in $\mathcal{O}'$ that dominates $O_i$. Color $O_i$ by a color among the 4 that is not used for $O_j$. Thus, we have colored all octants in $\mathcal{O}$ by 4 colors.
We prove the coloring is proper. 

Consider any hyperedge $e$ with $|e|\ge 2$ corresponding to a point $p$. Suppose $p$ is in $O_i \in \mathcal{O} \setminus \mathcal{O}'$. Then, for some $j \neq i$,  $p$ is also in $O_j \in \mathcal{O}'$. Since $O_i$ and $O_j$ are colored by distinct colors by our scheme, $p$ is contained in two octants of distinct colors. Now, suppose $p$ is in $O_i, O_j$ that are both in $\mathcal{O}'$. Note that $O_i \cap O_j$ contains all the points for which  $x \geq \max(a_i,a_j), y \geq \max(b_i,b_j), z \geq \max(c_i,c_j)$. Therefore, this intersection also is an octant $O_{ij}$ with apex at $\max(a_i,a_j)$, $\max(b_i,b_j)$, $\max(c_i,c_j)$. As $c_{\max} \geq c_{ij} = \max(a_i,a_j)+\max(b_i,b_j)+\max(c_i,c_j)$, the intersection of $O_{ij}$ with $x+y+z = c_{\max}$ is non-empty. Hence, the intersection of the projections of $O_i$ and $O_j$ which belong to $\mathcal{T}'$ 
is also non-empty, so the 
two projected triangles are colored by distinct colors.  This implies that $O_i$ and $O_j$ are colored by distinct colors. 
\end{proof}

\begin{theorem}
    Octants admit an extraction number of 4 and there are instances with octants where the minimum extraction number is 4. 
\end{theorem}
\begin{proof}
By Proposition~\ref{prop:coloring-to-extr-thm} and Lemma~\ref{lem:octant}, we obtain the Extraction Theorem with an extraction number of 4 for octants. 
We can construct an instance with octants where the minimum extraction number is 4. 
Consider a configuration of 4 octants whose projections on a plane $x+y+z=c$ for some $c$ are shown in Figure~\ref{fig:octants}. Moreover, all points lie on this plane.  
\begin{figure}
    \centering
    \includegraphics[width=0.3\linewidth]{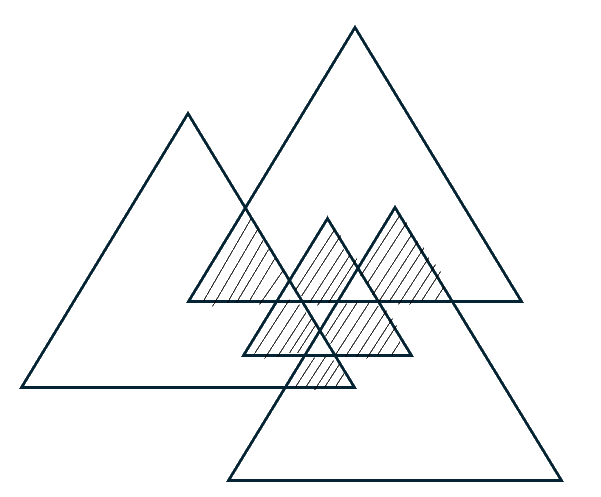}
    \caption{Configuration of projection of octants demonstrating extraction number $\geq 4$}
    \label{fig:octants}
\end{figure}
In this construction, every two triangles intersect pairwise and there exists a region (shown by $6$ shaded regions) which is the intersection of exactly two pairwise intersecting triangles. If there are points in each of these shaded regions, then $3$ out of $4$ octants will be needed to cover these points. Thus, the bound 4 is tight and the minimum extraction number of $4$ is achieved for this configuration. 
\end{proof}

\section{Conclusion}
In this work we proved Extraction Theorems along with the tightness of the bounds for classes of geometric objects with small extraction numbers. The proof of our theorems also yields  algorithms to compute the  geometric cover corresponding to these extraction numbers in polynomial time. 

Our work focused on classes of objects that include intervals, axis-parallel segments, axis-parallel rays, and octants. We investigate these classes of objects and prove small bounds on the extraction numbers. However, there is a vast class of objects for which extraction numbers are unknown or are not constant and depend on the size of the input. Some of these classes include unit disks in $\mathbb{R}^3$ and rectangles in $\mathbb{R}^2$. The mapping of the problem to that of hypergraph coloring and linking the extraction number to a proper $k$-coloring of the hypergraph  provides a framework to investigate these classes of objects. 

\paragraph{Acknowledgements} We are grateful to Shakhar Smorodinsky for pointing out the connection to hypergraph coloring. We also thank the reviewers for their valuable comments and suggestions. The work of Bandyapadhyay has been supported by the NSF grant 2311397. 
\bibliographystyle{plainurl}
\bibliography{main}

\end{document}